\documentclass[10pt,article]{IEEEtran}

\usepackage{amsmath}
\usepackage{amssymb}
\usepackage{epsfig}

\numberwithin{equation}{section}

\def\A{{\mathcal A}}
\def\B{{\mathcal B}}
\def\C{{\mathcal C}}
\def\D{{\mathcal D}}
\def\E{{\mathcal E}}

\def\G{{\mathcal G}}
\def\H{{\mathcal H}}

\def\S{{\mathcal S}}

\def\U{{\mathcal U}}
\def\V{{\mathcal V}}
\def\W{{\mathcal W}}

\def\NN{{\mathbb N}}
\def\FF{{\mathbb F}}

\newtheorem{theo}{Theorem}
\newtheorem{prop}{Proposition}
\newtheorem{lemm}{Lemma}
\newtheorem{coro}{Corollary}

\newcommand{\bey}[2]{\left[ {#1\atop #2} \right]}

  \date{}
\begin{document}

\title{On  generic erasure correcting sets and related problems}
\date{}
\author{ R. Ahlswede and H. Aydinian*
\thanks{This author was supported by the DFG Project AH46/8-1} \\
Department of Mathematics\\
University of Bielefeld\\
POB 100131,\\
D-33501 Bielefeld, Germany\\
ahlswede@math.uni-bielefeld.de\\
ayd@math.uni-bielefeld.de}
\maketitle

\vskip 1truecm

\begin{abstract}
Motivated by iterative decoding techniques for the binary erasure channel Hollmann and Tolhuizen
introduced and studied the notion of generic erasure correcting sets for  linear codes.
A generic $(r,s)$--erasure correcting set generates for all codes of codimension ~$r$~ a parity 
check matrix that allows iterative decoding
of all correctable erasure patterns of size $s$ or less. 
The problem is to derive
 bounds on the minimum size $F(r,s)$ of generic erasure correcting sets
and to find  constructions for such sets.
In this paper we continue the study of these sets.
We derive better lower and upper bounds.
Hollmann and Tolhuizen also introduced the stronger notion of $(r,s)$--sets
and derived bounds for their minimum size $G(r,s)$. Here also we improve these  bounds.
We  observe that these two conceps are closely related  
  to so called $s$--wise intersecting codes, an area, in which $G(r,s)$ has been  studied
primarily with respect to ratewise performance. We derive connections.
Finally, we observed that hypergraph covering can be used for
both problems to derive good upper bounds. 
\end{abstract}
\begin{keywords}
 Iterative decoding, stopping redundancy, generic erasure correcting set, intersecting code
\end{keywords}

\section{Introduction }

Iterative decoding techniques, especially when applied to low-density
parity-check  codes, have recently attracted a lot of attention.
It is known that the performance of
iterative decoding algorithms in case of a binary erasure channel 
depends on the sizes of the {\it stopping sets} associated with a collection of parity check equations
 of the code  \cite{DPTRU}.  
Let $H$ be a parity--check matrix of a code $\C$, defined as 
 a matrix whose rows span  the dual code $\C^\bot$.
 A stopping set 
 is a nonempty set of code coordinates such that the  submatrix formed by the
corresponding columns of $H$ does not contain
a row of weight one.  
Given a parity-check matrix $H$,
the size of the smallest nonempty stopping set, denoted by $s(H)$, is called the
{\it stopping distance} \cite{SV} of the code with respect to $H$.
Iterative decoding techniques, given a parity check 
matrix $H$, allow to correct all 
erasure patterns of size $s(H)-1$ or less. Therefore, for better
performance of iterative erasure decoding it is desired that $s(H)$ be as large as possible.
  Since the support of any codeword (the set of its nonzero coordinates)
is a stopping set, we have $s(H)\leq d(\C)$ for
all choices of $H$. It is well known  that the equality   
 can always be achieved, by choosing sufficiently many
vectors from the dual code $\C^\bot$ as rows in $H$. 
This motivated Schwartz and Vardy \cite{SV} to introduce the notion of
 {\it stopping redundancy} of a code.
The stopping redundancy
of $\C$, denoted by  $\rho(\C)$, is the minimum number of rows
in a parity-check matrix such that  $s(\C)=d(\C)$.\\
Schwartz and Vardy
\cite{SV} derived general upper and lower bounds, as well as
more specific bounds for Reed--Muller codes, Golay codes, and MDS codes. 
Improvements upon general upper bounds are presented in \cite{HS}, \cite{HSV}. 
The stopping redundancy of Reed--Muller codes was further studied by Etzion \cite{E}.
Hehn et al. \cite{HM} studied the stopping redundancy of cyclic codes.\\
Recall that a binary linear code $\C$ is capable of correcting those and only those  erasure patterns
that do not contain the support of a non-zero codeword. These patterns are
called {\it correctable} for $\C$. All other erasure patterns are called {\it uncorrectable}.
Note that  the size of  a correctable erasure pattern for a code can 
be greater than its minimum distance and it is upper bounded by the codimension of the code.\\
Hollmann and Tolhuizen  \cite{HT2} observed  that given a linear code $\C$,  
any correctable erasure pattern
can be iteratively decoded provided a chosen parity check matrix contains sufficiently 
many rows.
This motivated them \cite{HT2} to introduce the notion of {\it generic erasure correcting
 sets} for binary linear codes.
A generic $(r,s)$--erasure correcting set, {\it generic $(r,s)$--set} for short, 
 generates for all codes of codimension 
$r$ a parity check matrix that allows iterative decoding
of all correctable erasure patterns of size $s$ or less. 
More formally, a subset  $\A$  of a binary vector space $\FF^r_2$
 is called generic $(r,s)$--set if for
any  binary linear code $\C$ of length $n$ and codimension $r$, and any parity check
$r\times n$ matrix $H$ of $\C$, the set of parity check equations $\H_\A=\{{\bf a}H:{\bf a}\in \A\}$
enables iterative decoding of all correctable erasure patterns of size $s$ or less.\\
Weber and Abdel--Ghaffar \cite{WA} constructed parity check matrices for the Hamming code  
that enable iterative decoding 
of all correctable erasure patterns of size at most three. 
Hollmann and Tolhuizen  \cite {HT}, \cite{HT2}  gave a general construction
and established upper and lower bounds for the minimum size  of  generic $(r,s)$--sets. \\
Throughout the paper we use the following notation.
We use  $[n,k,d]_q$ for a linear code $\C$ 
(of length $n$, dimension $k$, and  minimum Hamming distance $d$)
over  ${\FF}_q$. The Hamminng weight of a vector ${\bf a}$ is denoted by $ wt({\bf a})$.
We denote by $[n]$ the set of integers $\{1,\ldots,n\}$. 
 A $k$--element subset of a given set
is called  for short a $k$--subset.
$\FF_q^{k\times m}$ denotes the set of all $k\times m$ matrices over
 the finite field ${\FF}_q$.
For integers $0\leq k\leq m$, $\bey {m}{k}_q$ stands for
the $q$-ary Gaussian coefficient, defined by
$\bey{m}{0}_q=1$~ and
~$\bey{m}{k}_q=\prod^{k-1}_{i=0}\dfrac{(q^{m-i}-1)}{(q^{k-i}-1)}$~ for
   $k=1,\ldots,m$.
It is well known that $\bey{m}{k}_q$ is the number of $k$--dimensional subspaces in $\FF^m_q$. 
A $k$--dimensional subspace is called for short a $k$--subspace.
A coset of a $k$--subspace in $\FF^m_q$ is called a $k$--dimensional plane or shortly
$k$--plane. Recall that there are $q^{m-k}\bey{m}{k}_q$ ~$k$--planes in $\FF^m_q$.
 A $k$--plane which is not a subspace is called a $k$--flat.  
 Later on we will  omit $q$ in the notation above for the binary case.

 In this paper we continue the study of generic erasure correcting sets.
Let $F(r,s)$ denote the minimum size of a generic $(r,s)$--set.
The bounds for $F(r,s)$ presented below are due to Hollmann and Tolhuizen.
The following   is the best known constructive bound
\begin{theo} \cite{HT2}\label{Thm1} For $2\leq s\leq r$ we have
\begin{equation}
F(r,s)\leq \sum_{i=1}^{s-1}\binom{r-1}{i}.
\end{equation}
\end{theo}
It is clear that any upper bound for  $F(n-k,d-1)$ 
is an upper bound for the stopping distance $\rho(\C)$ of an $[n,k,d]$ code, 
thus $\rho(\C)\leq F(n-k,d-1)$
Therefore, for an $[n,k,d]$ code $\C$  one has the bound
\begin{equation}
\rho(\C)\leq F(n-k,d-1)\leq \sum_{i=1}^{d-2}\binom{n-k-1}{i},
\end{equation}
which turns to be also the best constructive bound for the stopping redunduncy.\\
 We notice that the best known nonconstructive upper bounds for the stopping redundancy 
of a  linear code are given in 
Han and Siegel \cite{HS} and in Han et al \cite{HSV}.

\begin{theo} \cite{HS} \label{theo2} For  an $[n,k,d]$ code $\C$ with $r=n-k$ 
\begin{equation}
\rho(\C)\leq \min\{t\in \NN: \sum_{i=1}^{d-1}\binom{n}{i}
\Bigl(1-\frac{i}{2^i}\Bigr)^{t}<1\}+r-d+1.
\end{equation}
\end{theo}
A closed form expression derived from (I.3) is as follows

\begin{coro} \label{coro1} For  an $[n,k,d]$ code $\C$ with $r=n-k$ 
\begin{equation}
\rho(\C)\leq\frac{\log \sum_{i=1}^{d-1}\binom{n}{i}}{-\log\Bigl(1-\frac{d-1}{2^{d-1}}\Bigr)}+r-d+1
\end{equation}
\end{coro}
(where $\log$ is always of base 2). Further improvements upon the probabilistic upper  bound are given in \cite{HSV}.

There is a big gap between 
the lower and upper bounds for $F(r,s)$.

\begin{theo}\cite{HT}\label{theo3} For $1\leq s\leq r$ the following holds 
 \begin{equation}
  r\leq F(r,s)\leq\frac{rs}{-\log(1-{s}{2^{-s}})}.
 \end{equation}
\end{theo}
The upper bound is derived by a probabilistic approach.

In \cite{HT}  introduced and studied a  related notion of
 $(r,s)$-good set.\\
A subset $\A\subseteq {\FF}^r$
is called $(r,s)$-1 good if for any $s$ linearly independent 
vectors ${\bf v_1},\ldots,{\bf v_s}\in {\FF}^r_2$
there exists a vector ${\bf c}\in \A$ such that the inner product 
${\bf (c,v_j)}=1$ for $j=1,\ldots,s$.\\
Furthermore, 
$\A$ is called $(r,s)$-good if for any linearly independent vectors 
${\bf v_1},\ldots,{\bf v_s}\in {\FF}^r_2$ and for arbitrary $(x_1,\ldots,x_s)\in\{0,1\}^s$
there exists ${\bf c}\in \A$ such that ${\bf (c,v_j)}=x_j$ for $j=1,\ldots,s.$

We denote by $G_1(r,s)$ the minimum cardinality $|\A|$ for which there exists a 
$(r,s)$-1 good set $\A$.
The corresponding notation for $(r,s)$--good sets is $G(r,s)$.
Hollman and Tolhuizen  observed that  these two notions are essentially the same.
\begin{prop} \cite{HT} \label{prop1} Let  $\A\subseteq {\FF}^r$ be an $(r,s)$-1 good set, then $\A\cup \{{\bf 0}\}$
is an $(r,s)$-good set. Moreover, one has  $G_1(r,s)=G(r,s)-1.$
\end{prop}
Later on we consider only $(r,s)$-1  good sets
and call them for short just $(r,s)$--sets. 
Obviously every $(r,s)$--set is a generic $(r,s)$--set, thus $G_1(r,s)\geq F(r,s)$.
\begin{theo} \cite{HT}\label{theo4}. For $1\leq s\leq r$ the following holds
\begin{equation}
2^{s-1}(r-s+2)-1\leq G_1(r,s)\leq  \dfrac{rs-\log s!}{-\log(1-2^{-s})}.
\end{equation}
\end{theo}
The upper bound is obtained again by a probabilistic argument.

The  paper is organized as follows.\\
 In Section 2 we obtain some properties of generic $(r,s)$--erasure correcting sets 
and $(r,s)$--sets which we use later. \\
In Section 3 we show that 
the problem we study here is closely related to  so called $s$--wise intersecting
codes studied in the literature (\cite{CZ},\cite{CELS}). This  allows us to get more insight
about the problems mentioned above. \\
In Section 4 we focus on bounds 
 for $F(r,s)$ and $G_1(r,s)$. We improve  the  bounds (I.5) and (I.6) 
in Theorems 11--15.
In particular, we show that for $2\leq s< r$ we have
$$
3\cdot2^{s-2}(r-s)+5\cdot2^{s-2}-2\leq G_1(r,s)
\leq \frac{(r-s+1)s+2}{-\log (1-2^{-s})},$$
$$
 F(r,s)>\max \{2^{s-1}+r-s,G_1(r-\lceil{s}/{2}\rceil,\lfloor{s}/{2}\rfloor)\},$$ 
$$F(r,s)<\dfrac{rs-\log s!}{-\log(1-s2^{-s})}.
$$
%In particular, the upper bound  for $G_1(r,s)$ allows to improve
 %the lower bound in \cite{CZ} for the rate of  $s$--wise intersecting codes.\\   
In Section 5 we show  that  hypergraph covering 
can be used  to obtain in a simple way good upper bounds for generic erasure correcting sets,
 $(r,s)$--sets,
and stopping redundancy of a linear code.

\section{Properties of generic $(r,s)$--sets }

Hollmann and Tolhuizen obtained the following characterization
of generic $(r,s)$--sets.
\begin{prop}\cite{HT2} \label{prop2}
 A subset $\A\subset \FF^r$ is generic $(r,s)$--set if and only if for every
full rank matrix $M\in \FF^{r\times s}$ there exists ${\bf a}\in \A$ such that 
$wt({\bf a}M)=1$. 
\end{prop}
We extend this characterization as follows
\begin{prop}\label{prop3}
 A subset $\A\subset \FF^r$ is a generic $(r,s)$--set if and only if
for every full rank matrix $M\in\FF^{r\times s}$ the set 
$\{{\bf x}={\bf a}M: {\bf a}\in \A \}\subset \FF^s$
contains a hyperlane  not passing through the origin.
\end{prop}
\begin{proof} For  integers $1\leq t\leq s<r$ and a set of linearly 
independent vectors 
$S=\{\bf v_1,\ldots,v_t\}\subset \FF^s$, let   $\A\subset \FF^r$ be a subset satisfying
 the following property 
with respect to $\{\bf v_1,\ldots,v_t\}$:\\
(P)~ For every full rank matrix $M\in \FF^{r\times s}$ 
 there exists a vector ${\bf a}\in \A$
such that ${\bf a}M={\bf v_i}$ for some $i\in [t]$.

 We claim then that $\A$ satisfies this property 
with respect to every linearly independent set  of vectors
 $\{\bf x_1,\ldots,x_t\}\subset \FF^s$. 

To prove the claim, we have to show that given a full rank matrix $M\in \FF^{r\times s}$,
there exists ${\bf a}\in \A$ such that   ${\bf a}M={\bf x_i}$ for some $i\in [t]$.
 Let  $N\in \FF^{s\times s}$ be an invertible  matrix
such that ${\bf v_i}N={\bf x_i}$ for $i=1,\ldots,t$. Then, in view of the property (P) of $\A$,  there exists 
${\bf a}\in \A$ such that ${\bf a}(MN^{-1})={\bf v_i}$ for some $i\in[t]$
 and hence ${\bf a}M= {\bf v_i}N={\bf x_i}$.\\
Let now $t=s$ and let $S$ be the set of $s$  
unit vectors in $\FF^s$. Then the claim (together with Proposition \ref{prop2}) gives the following
 analogue of Proposition \ref{prop2}.

 \begin{prop}\label{prop4} A set $\A\subset \FF^r$ is generic $(r,s)$--set if and only if for any given set
of linearly independent vectors 
$\{\bf v_1,\ldots,v_s\}\subset \FF^s$ and every
full rank matrix $M\in \FF^{r\times s}$ there exists ${\bf a}\in \A$ 
such that ${\bf a}M={\bf v_i}$ for some $i\in [s]$. 
\end{prop}

Note  also that for $|S|=t=1$ we have $(r,s)$--sets  and the claim implies the following condition (shown in \cite{HT}):  
$\A\subset \FF^r$ is an $(r,s)$--set if and only if
for every full rank matrix $M\in\FF^{r\times s}$ the set 
$\{{\bf x}\in \FF^s:{\bf x}={\bf a}M, {\bf a}\in \A \}$
contains all nonzero  vectors.
This condition clearly means that $\A$ meets every $(r-s)$--flat.
 
Let now $\A$ be a generic $(r,s)$--set and let 
 $M\in \FF^{r\times s}$ be a matrix of rank $s$. Let also 
${\bf u_1,\ldots,u_s }\in \FF^s$ be such that
 $\{{\bf a}M: {\bf a}\in \A\}\cap \{\bf u_1,\ldots,u_s\}=\emptyset$.
Then Proposition 4 implies that the dimension dimspan$\{{\bf u_1,\ldots,u_s}\}\leq s-1$.
Thus, $\FF^s\setminus$span$\{\bf u_1,\ldots,u_s\}$ contains a hyperplane
not passing through the origin. 

Furthermore, suppose that for every full rank matrix $M\in \FF^{r\times s}$ there exists an $(s-1)$--flat 
$\U\subset \{{\bf a}M: {\bf a}\in \A\}$.  Note then that for every linearly independent vectors ${\bf u_1,\ldots,u_s}\in \FF^s$
we have $\{ {\bf u_1,\ldots,u_s}\}\cap \U\neq \emptyset$. This, in view of Proposition 4, implies that $A$   
is a generic $(r,s)$--set.
\end{proof}

Let $\A\in \FF^r$ be a generic $(r,s)$--set.
Let us represent $\A$ by an $|\A|\times r$ matrix $A$ where the rows are the vectors of $\A$.
Let also $N\in \FF^{r\times r}$ be an invertible matrix.
Then we get the following.
\begin{coro}\label{coro2}
(i) In every set of  $s$ columns of $AN$ there is a subset of $s-1$ columns
 that contains each $(s-1)$--tuple. \\
(ii) $\A$ hits  at least $2^{s-1}\bey r{r-s}$ ~  $(r-s)$--flats. \\
(iii) $|\A|\geq 2^{s-1}+r-s$.
\end{coro}
\begin{proof} (i) Note first that
the rows of $AN$ also define a generic $(r,s)$--set. Indeed, in view of Proposition \ref{prop2}
for every full rank matrix $M\subset \FF^{r\times s}$ (and hence for $NM$)
  the matrix $A(NM)=(AN)M$ contains a  row of weight one. 
Now the statement  follows from Proposition \ref{prop3}. \\
(ii) Proposition \ref{prop3} implies that if $\A\subset \FF^r$ is a generic $(r,s)$--set,
then $A$ hits at least $2^{s-1}$ cosets of every $(r-s)$--subspace in $\FF^r$. 
This implies the statement.\\
(iii) Without loss of generality we may assume that $A$ contains
 $r$  unit vectors. Now the statement follows since
there exist $s-1$ columns of $A$ that contain all nonzero $(s-1)$--
tuples and $r-s+1$ zero tuples.
\end{proof}

\section {Relation to other problems}
In this section we show the relationship between $(k,s)$--sets and
$s$--wise intersecting codes

{\it Intersecting Codes}: A linear $[n,k]_q$ code $\C$  over a field $\FF_q$ 
is called {\it intersecting} if any two 
nonzero codewords have a common nonzero coordinate.
 Intersecting codes where introduced in \cite{KS} and have been studied by several authors 
\cite{KS}, \cite{M}, \cite{CL}, \cite{CELS}, \cite{CZ}.

 A more general notion of {\it $s$-wise intersecting} codes was introduced in \cite{CL}.
A set of vectors $A\subset \FF_q^n$ is called $s$--wise intersecting if there is a coordinate where all the vectors have a nonzero element.\\
An  $[n,k]_q$ code  is called  $s$-wise intersecting ($s\geq 2$) if 
 every subset of  $s$ independent vectors in it is $s$--wise intersecting.

{\it Problem 1}
Given integers $2\leq s\leq k$,  determine $n_q(k,s)$ (in case $q=2$ we write $n(k,s)$),
 the minimum length $n$ of an
 $s$-wise intersecting $[n,k]_q$--code.
\begin{prop} \label{prop5}
The elements of a $(k,s)$--set $A\subseteq \FF^k_2$ with $|A|=n$, represented as columns of a matrix, give a generator matrix of 
 an $s$-wise intersecting $[n,k]$ code. Conversely, the columns of a generator matrix of an  $s$-wise intersecting $[n,k]$ code form
a  $(k,s)$--set.
As a consequence we have
$G_1(k,s)=n(k,s)$.
\end{prop}
\begin{proof}.
Let $\A=\{{\bf a_1},\ldots,{\bf a_n}\}\subseteq \FF_2^k$ be a $(k,s)$--set.
Let us represent $\A$ as an $n\times k$ matrix  $A$ where  the rows correspond 
to the vectors of $\A$,
 and denote $G=A^T$. Note that  $G\in \FF^{k\times n}_2$ and $rank(G)=k$. 
Let ${\bf v_1},\ldots,{\bf v_s}\in {\FF_2}^k$ be linearly independent vectors and let
${\bf u_1}={\bf v_1}G,\ldots, {\bf u_s}={\bf v_s} G$.
 Then  ${\bf u_1},\ldots,{\bf u_s}\in {\FF}^n_2$
are linearly independent as well. By the definition of a $(k,s)$--set, 
there exists ${\bf{a_i}}\in \A$ such 
that ${\bf (a,v_j)}=1$ for $j=1,\ldots,s$, that is all vectors 
${\bf u_1},\ldots,{\bf u_s}$ have a one in the $i$th coordinate.
This clearly means that the $[n,k]$ code with  generator matrix 
$G$ is an $s$--wise intersecting code.
Similarly we have the inverse implication.
 \end{proof} 

Recall (Proposition \ref{prop1}) that  if $\A\subseteq \FF_2^k$ is a $(k,s)$--set then $\A$  contains a solution
to every  (consistent) nonhomogeneous system of  $s$ independent equations, which in fact  means that $A$ meets every $(k-s)$--flat.
Thus,  the problem of construction of 
 $s$--wise intersecting $(n,k)$--codes (respectively $(k,s)$--sets) can be viewed as a covering problem.  

 {\it Problem 2} Determine the minimal size $n(k,s)$ of a set of vectors in  $\FF^k$, 
called  a transversal or a blocking set, that meets every $(k-s)$--dimensional flat.

{\it Remark 1} We note that in case $s=1$ we have a triviality and $n(k,1)=k$.
Another trivial case is $s=k$. In this case we clearly have $n(k,k)=2^k-1$.

Also it is not hard to observe  that $n(k,k-1)=2^k-2$ (see also Remark 3 below). The first open case is $s=2$.

{\it Remark 2} The notion of a $(k,s)$--set can be extended to arbitrary spaces $\FF^k_q$
in a natural way.
 However, notice that Proposition \ref{prop5} is not true for the nonbinary case. 
Consider an MDS $[n,k,d=n-k+1]_q$--code $\C$.
Such a code exists for all $1\leq k\leq n\leq q+1$ (see \cite{MS}). Observe  that for
 $d> \frac{s-1}{s}\cdot n$ (that
is $n> s(k-1)$)
we have an $s$--wise intersecting code, but  the columns of a generator matrix
 of $\C$ do not form
a $(k,s)$--set for $s\geq 2$.

 It is  worth to mention that the problem of finding  the minimal size  of a set
 of nonzero vectors in $\FF^k_q$
that meets all $(k-s)$--dimensional subspaces  is much easier.
This problem  was solved  by Bose and Burton \cite{BB}.

\begin{theo} \cite{BB}  \label{theo5}~Let $\A$ be a set of points of $\FF^k_q$ that  meets every 
$(k-s)$--space of $\FF^k_q$. Then $|\A|\geq (q^{r+1}-1)/(q-1)$, with equality if and only if
$\A$ consists of  the points of an $(r+1)$--subspace of $\FF^k_q$. 
\end{theo}

{\it Covering arrays}: A  $k\times N$ array with entries from an alphabet of size
 $q$ is called a $t$-covering array,
and denoted by CA$(N,k,t)_q$, if
the columns of each 
$t\times N$ subarray contain each $t$-tuple  at least once as a column. 
The problem is to minimize $N$ for which there exists a CA$(N,k,t)_q$.
Covering arrays were first introduced by Renyi \cite {Ren}. The case $t=2$ was  solved
by  Renyi \cite {Ren} (for even $k$) and  by Katona \cite{K} and 
Kleitman and Spencer \cite{KlS}
(for arbitrary $k$).
Covering arrays have applications in circuit testing, digital communication,
 network designs, etc.
Construction of optimal covering arrays has been the subject of a lot of research
(see a survey \cite{C}).

Let $G$ be a generator matrix of an $s$--wise intersecting $[n,k]$ code $\C$ and let $M\in \FF^{s\times k}$ be
a full rank matrix. Then in view of Proposition 5 (and by definition of an $(s,k)$--good set) the columns
of matrix $MG$ contain all nonzero $s$--tuples. This  in particular means that for every invertible matrix $L\in\FF^{k\times k}$
the matrix  $G^{\prime}=LG$ (a generator matrix of $\C$) together with the all zero column
is a covering array. 
 Thus, we have the following. 
\begin{prop} \label{prop6} 
An $[n,k]$ code  $\C$ is $s$--wise intersecting if and only if 
 every generator matrix of $\C$ (together with the all zero column) is an 
$s$--covering array.\\
 Equivalently, the columns  of an $s$--covering $k\times N$ array CA over binary alphabet
(considered as vectors in $\FF^k$) form an $(k,s)$--good set 
if and only if CA is invariant under every invertible 
transformation of  $\FF^k$.  
\end{prop}
Let us also mention another extensively studied related notion. 
 A  code $\C$  of length $n$ is called
 $(t,u)$--{\it separating}, if for every disjoint pair $(U,T)$ of
subsets of $\C$ with $|T|=t$ and $|U|=u$  the following holds: there exists 
a coordinate $i$ such 
that for any codeword $(c_1,\ldots,c_n)\in T$ and any codeword  $(c'_1,\ldots,c'_n)\in U$,
$c_i\neq c'_i$. \\
 Separating codes were
  studied  by many authors in connection with practical problems
in cryptography, computer science, and search theory.
The relationship between $s$--wise intersecting codes and separating codes is 
studied in \cite {CELS}.

\subsection{Some known results about intersecting codes}
In this subsection we present some known results on intersecting codes
which can be used for our problems. \\
 Given a vector ${\bf v}=(v_1,\ldots,v_n)\in \FF^n_q$, 
the set $I=\{i\in [n]: v_i\neq 0\}$ is called the support  of ${\bf v}$  and is
denoted  by   $supp(\bf v)$. Given a code $\C$ of length $n$ 
 and $I= \{i_1,\ldots,i_{|I|}\}$, denote by $\C(I)$ the restriction 
of the  code  on the coordinate set $I$, that is the code obtained by deletion of 
the coordinates $\bar I\triangleq\{1,\ldots,n\}\setminus I$.

\begin{lemm} \label{lemm1} Let $\C$ be an $s$-wise intersecting $[n,k]$ code and let $\bf v\in \C$
be a codeword with $wt({\bf v})=w$ and with
$supp{(\bf v)}=I$.  Then\\
(i) \cite{CELS} ~$\C(I)$ is an $[w,k]$-code.  
 If $\bf \{u_1,\ldots,u_{k-1},v\}$ is  a base of $\C(I)$ then the code  $\C^*(I)$
 generated by the vectors $\bf \{u_1,\ldots,u_{k-1}\}$
is an $(s-1)$--wise intersecting $[w,k-1]$  code.\\
(ii) $\C(\bar I)$ is an $(s-1)$--wise 
intersecting $[n-w,k-1]$ code.
\end{lemm}

The proof of (i) is easily derived from the definition of an $s$--wise intersecting code.
Note that both (i) and (ii) follow from Proposition \ref{prop6}
(the lemma was also observed in \cite{HT} in terms of $(r,s)$--sets). 

 Lemma 1 implies simple estimates for the minimum and maximum distances of intersecting codes.
It shows that $s$-wise intersecting codes have strong distance
 properties which means that in general
construction of  such optimal codes is a difficult problem.

In view of equivalence shown in Proposition \ref{prop5}, the next  results  can be used for construction of infinite families of 
$(r,s)$--sets
with positive rate.
\begin{theo} (Cohen--Zemor) \cite{CZ} \label{theo6} There is a constructive infinite sequence of $s$-wise
 intersecting binary codes
with rate arbitrary close to 
\begin{equation}
R=\Bigl(2^{1-s}-\frac{1}{2^{2s+1}-1}\Bigr)\frac{2s+1}{2^{2s}-1}=2^{2-3s}(s+o(s)).
\end{equation}
\end{theo} 
The  result is  obtained by concatenating algebraic-geometric $[n,k,d]_q$ codes in Tsfasmann \cite{T}
satisfying $d>n(1-2^{1-s})$ with $q=2^{4s+2}$ and with a rate arbitrary close 
to $2^{1-s}-1/(\sqrt q-1)$,
with $s$--wise intersecting $[2^{2s+1}-2,4s+2,2^{2s}-2^s-1]$ code (the  punctured dual of the 2-error-correcting BCH code).

Another possible approach for constructing $s$--wise intersecting codes (and hence $(r,s)$--sets) is to use $\varepsilon$-{\it Biased Codes}.
A binary linear code $\C$ of length $n$ is called $\varepsilon$--biased if the weight of every non-zero codeword in $\C$ 
lies in the range 
$(1/2-\varepsilon)n\leq w\leq (1/2+\varepsilon)n$. Biased codes can be constructed 
using pseudo-random graphs 
known as expanders (expander  codes).

\begin{theo} (Alon et al.) \cite{ABNNR} \label{theo7} For any $\varepsilon>0$, there exists an explicitly
 specified family of constant-rate binary linear $\varepsilon$--biased codes.
\end{theo}

\begin{lemm} (Cohen--Lempel) \cite{CL} \label{lemm2}
Let $d$ and $D$ denote respectively the minimum and the maximum distance of
 a binary linear code  $\C$. Then $\C$ is $s$--wise intersecting if $d>D(1-2^{1-s})$. 
\end{lemm}

The next statement  follows directly from Lemma 2.

\begin {coro} \label{coro3}
 An $\varepsilon$--biased linear code is $s$--wise intersecting if 
$\varepsilon<{1}/(2^{s+1}-2)$.
\end {coro}

The following nonconstructive lower  bound for the rate of an $s$--wise intersecting
$[n,k]$ code is due to Cohen and Zemor.

\begin{theo} \cite{CZ} \label{theo8}
 For any given rate $R<R(s)$
\begin{equation}
 R(s)=1-\frac{1}{s}\log(2^s-1)
\end{equation}
and $n\rightarrow \infty$ there exists an $s$--wise intersecting $[n,k]$ code of rate $R$.
\end{theo}

Using recursively the upper bound  due to McEliece-Rodemich-Rumsey-Welch \cite{MS} together
with Lemma 1 (i) one can get upper bounds for the rate of $s$--wise intersecting codes.
\begin{theo} (Cohen et al.) \cite{CELS} \label{theo9}
 The asymptotic rate of the largest $s$--wise intersecting code
is at most $R_s$, with $R_2\approx 0.28,~ R_3\approx 0.108, R_4\approx0.046, R_5\approx 0.021, R_6\approx 0.0099$. 
\end{theo}

For the case $s=2$, the best known bounds 
 on the minimal length $n(k,2)$ of an $[n,k]$ intersecting code are as follows
\begin{equation}
c_1(1+o(1))k< n(k,2)< c_2k-2,
\end {equation}
where $c_1=3.53\ldots, c_2=\frac{2}{2-\log3}$.

The lower bound is obtained by Katona and  Srivastava \cite{KS}. 
The upper bound is due 
to Koml\'{o}s (see \cite{KS}, \cite{M},  \cite{CL}).
Note that  the upper bound in Theorem 4  for $s=2$ 
gives 
$G_1(k,2)=n(k,2)\leq \frac{2k-1}{2-\log 3}<\frac{2}{2-\log 3}\cdot k-2.$

\section{Improving  bounds for $G(k,s)$ and $F(k,s)$}

In this section we derive new bounds for $G_1(k,s)$ and $F(k,s)$. We first derive a lower bound for $G_1(k,s)$.
Recall that we have trivial cases $G_1(k,1)=n(k,1)=k$ and  $G_1(k,k)=n(k,k)=2^s-1.$ 
\begin{theo} \label{theo10}For $2\leq s\leq k-1$ we have
\begin{equation}
G_1(k,s) \geq 3\cdot2^{s-2}(k-s)+5\cdot2^{s-2}-2.
\end{equation}
\end{theo}
\begin{proof} To prove this bound we need the following consequence of Lemma 1.
\begin{lemm} \label{lemm3}
 For an $s$--wise intersecting $[n,k,d]$ code $\C$ with maximum distance $D$ we have
\begin{equation} 
n\geq 2 \cdot n(k-1,s-1)+D-d+1.
\end{equation}
\end{lemm}
\begin{proof} Let ${\bf v}$ be a codeword of minimal weight $d$,  
with the support set $I$, that is $wt({\bf v})=|I|=d$, and 
let $G$ be a generator matrix of  $\C(I).$
  We may assume that all rows of $G$ except for the first one have 
a zero in the first coordinate. Hence by Lemma 1(i) the 
  code $\C^*(I)$ (defined in Lemma 1)  has  support size $d-1$, that is $d\geq n(k-1,s-1)+1.$
Furthermore, Lemma 1(ii) implies  that $ D\leq n-n(k-1,s-1)$, which together with
the previous inequality gives the result. 
 \end{proof}

Recall now that for $s<k$ we have $n(k,s)<2^k-1$. Note then that $D>d$.
 This follows from the simple observation
that there is no a constant weight $[n,k,d]$ code with $n<2^k-1$. 
Then  Lemma \ref{lemm3} in particular implies the inequality 
$n(k,s)\geq 2n(k-1,s-1)+2$
(the latter also follows from the fact that in case $n(k,s)<2^k-1$ we have $n\geq 2d$).
Since $\C^*(I)$ is an $[d,k,d^\prime]$ code, there is 
a codeword ${\bf u}\in \C$ of weight at most $d^\prime$ in the support set $I$ of $\bf v$.
Observe that this  implies $2d-2d^\prime\leq D\leq n-n(k-1,s-1)$ and
hence $n\geq n(k-1,s-1)+2d-2d^\prime$,
where $d^\prime$ is the minimum weight of $\C(I)$. 
Note that $d^\prime \leq d-k+1$ and thus   $n\geq n(k-1,s-1)+2k-2$.
This in particular for $s=2$ (together with $n(k-1,1)=k-1$) implies that
 $n(k,2)\geq 3k-3$.
 We have now the relation
$$G_1(k,s)\geq 2G_1(k-1,s-1)+D-d+1$$
\begin{equation}\geq 2G_1(k-1,s-1)+2
\end{equation}
 with  $G_1(k,2)=n(k,2)\geq 3k-3$.
Using induction on $s\geq 2$ we get  the required result.
\end{proof}
Notice that the right hand side
of (IV.1) is greater than the lower bound $2^{s-1}(k-s+2)-1$  in (I.6) $(k=r)$
by $2^{s-2}(k-s+1)-1$. Note also that this lower was obtained (in \cite{HT})  using 
the relation $G_1(k,s)\geq 2G_1(k-1,s-1)+1$ (compare with Lemma 3, resp. with (IV.3)).

{\it Remark 3}: 
The  bound  (IV.1) is tight for $s=k-1$. Indeed, we have
 $G(k,k-1)\geq 3\cdot 2^{k-3}+5\cdot 2^{k-3}-2=2^k-2.$
On the other hand any set of $2^k-2$ nonzero vectors is a $(k,k-1)$--set. 
The latter $(k-1)$--wise intersecting  $[2^k-2,k]$ code is a 
punctured simplex code. 
\begin{theo} \label{theo11} For $2\leq s< k$ we have \\
$G_1(k,s)\leq$ 
\begin{equation} \min_{N\in\NN} \Bigl\{N: \prod_{j=1}^{N}\Bigl(1-\frac{2^{k-s}}{2^k-j}\Bigr)
 (2^s-1)\bey ks<1\Bigr\}.
 \end{equation}
\end{theo}
\begin{proof} Our problem is to find a blocking set of (minimum) size $N$ with respect to 
the $(k-s)$-dimensional flats in $\FF^k_2$. Let
$U$ be a $(k-s)$--flat and let $B=\FF^k_2\setminus U$.
The subset $B$ with $|B|=2^k-1-2^{k-s}$ does not contain a blocking set.
Thus, for every fixed $U$ there are
$\binom {2^k-2^{k-s}}{N}$ bad $N$--sets ($N$--sets which are not blocking sets) in $B$.
The number of all $(k-s)$--flats is $(2^s-1)\bey k{k-s}$. Therefore,
the number of bad sets of size $N$ is less than $\binom{2^k-1-2^{k-s}}{N-k}(2^s-1)\bey k{k-s}$.
 If now 
$ \binom {2^k-1-2^{k-s}}{N}(2^s-1)\bey k{k-s}<\binom{2^k-1}N$ (the number of all
$N$--subsets of $\FF^k_2\setminus\{\bf 0\}$) 
then there exists a blocking set of size $N$.
The latter inequality is equivalent to the following
\begin{equation}
\prod_{j=1}^{N}\Bigl(1-\frac{2^{k-s}}{2^k-j}\Bigr)(2^s-1)\bey{k}{s}<1.
\end{equation}
This gives the result.\end{proof}

Note that Theorem 11 improves the upper bound in Theorem 4. A closed form expression derived from (IV.4) is 
as follows.
\begin{coro} \label{coro4} For $2\leq s< k$ we have
\begin{equation}
G_1(k,s)< \frac{(k-s+1)s+2}{-\log (1-2^{-s})}. \label{}
\end{equation}
\end{coro}
\begin{proof} We use the following known estimate for the Gaussian coefficients which
is not hard to verify:
$\bey nm< 2^{m(n-m)}\prod_{i=1}^{m} \frac{1}{(1-2^{-i})}<2^{m(n-m)+2}.
$
The left hand side  of (IV.5) is less than 
$\Bigl(1-\frac{2^{k-s}}{2^k}\Bigr)^N2^{s(k-s+1)+2}$.
The latter implies that $N\geq \frac{(k-s+1)+2}{-\log(1-2^{-s})}$, hence the result.
 \end{proof}

 Corollary 4 in terms of the rate of an $s$--wise intersecting code gives the following

\begin{coro} \label{coro5}
 Given integers  $2\leq s<k$, there exists an $s$--wise 
intersecting $[n,k]$ code of rate 
\begin{equation}
 R>\frac{k}{k-s+2}(1-\frac{1}{s}\log(2^s-1))
\end{equation}
(compair with Theorem 8).
\end{coro}
\begin{proof} Denote $g(k,s)$ the right hand side of (IV.4) and $R(s)$ is defined as in Theorem 8. Note  then that
$$-\log (1-{2^{-s}})=s(1-\frac{1}{s}\log(2^s-1))=sR(s).$$ Therefore, in view of Corollary 4, we have 
$$R>\frac{k}{g(k,s)}=\frac{ks}{(k-s+1)s+2}\cdot R(s)\geq$$
$$ \frac{k}{k-s+2}\cdot R(s).$$
 \end{proof}

Next we derive bounds for $F(k,s)$. We start with a lower bound. Recall that in view of Corollary \ref{coro2}(iii) 
we have $F(k,s)\geq 2^{s-1}+k-s$, which actually improves the lower bound $F(k,s)\geq k$ (Theorem 3).
However, we are able to improve this bound.
\begin{theo} \label{theo12} For integers~  $4\leq s\leq k-1$~ and $t\in \NN$ we have
\begin{equation}
 F(k,s)\geq \min_{2\leq t\leq s} \max\{G(k,t-1), G(k-t,s-t)\}.
\end{equation}
\end{theo}

\begin{proof} Let $\A\subset \FF^r$ be a generic $(k,s)$--set with $|\A|=N$
 and let ${A}\in \FF^{k\times N}$ 
be a matrix where the columns are the vectors of $\A$. Denote by $\C\subset \FF^N$ 
 the $[N,k]$ code generated by ${A}$.
Suppose  that $2\leq t\leq s$ is the smallest number such that there exists a subset $B\subset\C$
of $t$ linearly
 independent vectors  
which is not $t$--wise intersecting. 
Thus  $\C$ is $(t-1)$--wise intersecting 
but not $t$--wise intersecting. Let also $B=B^{\prime}\cup \{\bf a\}$ where $B^{\prime}$ is an $(s-1)$--wise intersecting subset.
Without loss of generality, we  assume that the rows of $A$ contain
the vectors of $B$.  Denote then
by  $A^\prime$ the $(k-t)\times N$ submatrix of $A$ obtained after 
removing all vectors of $B$. 
  
We claim now that  the code $\C^\prime$ generated by $A^\prime$
is an $(s-t)$--wise intersecting  $[N,k-t]$  code (to avoid a triviality we assume that $s-t\geq 2$). 

Suppose this is not the case, and let  $D\subset\C^\prime$ be a set of  $s-t$ 
linearly independent vectors which are not $(s-t)$--wise intersecting. 
 Recall that, in view of Corollary \ref{coro2}(i) (and Proposition \ref{prop5}), every subset of $s$
linearly independent vectors in $\C$ contains an $(s-1)$--wise intersecting subset. 
Thus $B\cup D$ contains an $(s-1)$--wise intersecting subset $E\subset (B\cup D)$.
Furthermore $E$ contains one of subsets $B$ and $D$. 
Note however, that $B\nsubseteq E$ since $B$ is not $t$--wise intersecting ($2\leq t<s-1)$.
Similarly $D\nsubseteq E$, since (by assumption) $D$ is not $(s-t)$--wise intersecting ($2\leq s-t<s-1$).
This means that
set $B\cup D$  does not contain an  $(s-1)$--wise intersecting subset,
a contradiction.  Therefore, given $2\leq t\leq s$, we have 
$F(k,s)\geq  \max\{G(k,t-1), G(k-t,s-t)\}$ which completes the proof. 
\end{proof}

\begin{coro} \label{coro6} Given integers~   $4\leq s\leq k-1$ we have
 \begin{equation} 
  F(k,s)\geq G(k-\lceil{s}/{2}\rceil,\lfloor{s}/{2}\rfloor).
 \end{equation}
\end{coro}
\begin{proof}  
Note first that we have $G(k-t,s-t)\geq G(k-\lceil{s}/{2}\rceil,\lfloor{s}/{2}\rfloor)$ for any
$1\leq t\leq \lceil s/2\rceil$. In case $t> \lceil s/2\rceil$
we have $G(k,t-1)>G(k-\lceil{s}/{2}\rceil,\lfloor{s}/{2}\rfloor)$.
 This clearly implies that
 $\min_{2\leq t\leq s} \max\{G(k,t-1), G(k-t,s-t)\}
\geq G(k-\lceil{s}/{2}\rceil,\lfloor{s}/{2}\rfloor)$.
\end{proof}

To apply  Corollary 6 we can use any lower bound for $G(k,s)$. Using for example (I.6) we get
$F(k,s)\geq G(k-\lceil{s}/{2}\rceil,\lfloor{s}/{2}\rfloor)
\geq 2^{\lfloor\frac{s}{2}\rfloor-1}(k-s+2).$
Thus, we have
\begin{equation}
 F(k,s)\geq \max\{2^{s-1}+k-s,2^{\lfloor\frac{s}{2}\rfloor-1}(k-s+2)\}.
\end{equation}
For $s=4$ Corollary 6 together with (IV.I) implies
\begin{equation}
F(k,4)\geq G(k-2,2)\geq 3(k-3).
\end{equation}
\begin{theo} \label{theo13} For integers $2\leq s< k$ we have
$F(k,s)\leq $
\begin{equation} \min_{N\in\NN} \Bigl\{N: \prod_{j=1}^{N}\Bigl(1-\frac{s2^{k-s}}{2^k-j}\Bigr)
\frac{1}{s!}{\prod_{i=0}^{s-1}(2^s-2^i)} \bey ks<1\Bigr\}.
 \end{equation}
\end{theo}
\begin{proof}  To each $(k-s)$--subspace  $U\subset\FF^k$ we put 
into correspondence a fixed
generator  matrix $H\in \FF^{s\times k}$ of the dual space $V^\bot$, that is 
 $U=\{{\bf x}\in \FF^k:{\bf x}H^T={\bf 0}\}$. For example, taking the
set of  all $s\times r$ matrices of rank $s$ in reduced row echelon form, 
we get one--one correspondence
between these matrices and the set of all $(k-s)$--subspaces of $\FF^k$. 
 Now each coset  of $U$ denoted by $U_b$ is uniquely defined by the pair $(H,{\bf b})$ 
where ${\bf b}\in\FF^s$ and
$U_b=\{{\bf x}\in\FF^r: H{\bf x}^T={\bf b}^T\}$.
 We say that the cosets $U_{b_1},\ldots,U_{b_t}$ are 
linearly independent if the vectors ${\bf b_1,\ldots,b_t}$ are linearly independent.
Let $\B(U)$ denote the set of all cosets of $U$. We look  
for an $N$--subset of $\FF^k$ which is a generic $(k,s)$--set.\\
In view of Proposition 4, a subset $A\in\FF^r$ is a generic $(k,s)$--set iff
for each $(k-s)$--subspace $U$,  it contains a vector from every collection of $s$
linearly independent cosets of $U$. We  estimate now the number of bad sets of size $N$.
We remove from $\B(U)$ a set of  
$s$ independent cosets and denote the union of these cosets by $\S$, thus $|\S|=s2^{k-s}$.  
Then  any $N$--subset  of 
 $\FF^k\setminus \S$ is a bad set. The same holds with
 respect to the cosets of every $(k-s)$--subspace.
The number of distinct bases
 in $\FF^s$ is $\frac{1}{s!}{\prod_{i=0}^{s-1}(2^s-2^i)}$.
Therefore, the number of all bad $N$--subsets is less than 
$\binom{2^k-1-s2^{k-s}}{N}\frac{1}{s!}{\prod_{i=0}^{s-1}(2^s-2^i)}\bey k{k-s}$.
If now this number is less than $\binom {2^k-1}N$, the number of all $N$--subsets
 of $\FF^k\setminus \{\bf 0\}$,
then there exists a generic $(k,s)$--set of size $N$. The latter is equivalent to
 \begin{equation}
\prod_{j=1}^{N}\Bigl(1-\frac{s2^{k-s}}{2^k-j}\Bigr)\frac{1}{s!}
{\prod_{i=0}^{s-1}(2^s-2^i)} \bey ks<1.
 \end{equation}
This implies the result. 
\end{proof}
A closed form expression derived from (IV.12 ) is as follows.
\begin{coro} \label{coro7} For $2\leq s< k$ we have
 \begin{equation} 
F(k,s)<\frac{sk-\log s!}{-\log(1-\frac{s}{2^s})}.  
 \end{equation}
\end{coro}
\begin{proof} Simple calculations show that the left hand side 
of (IV.13) is less than $(1-\frac{s}{2^s})^N2^{sk}/s!$.
\end{proof}

\section{Bounds derived by a  hypergraph covering}

In this section we show, that  hypergraph covering
 can be employed to get good upper bounds for $(r,s)$--sets,
 generic erasure correcting sets, and  stopping redundancy of
a linear code. Recall that a hypergraph is a pair $(\V,\E)$ where $\V$ is a set of elements called vertices
and $\E$ is a set of nonempty subsets  of $\V$ called edges.  
Let $\H=(\V,\E)$ be a hypergraph with a vertex set $\V$ and an edge set $\E$.
We denote  by $d_\V=\min_ {v\in \V} deg(v)$ (minimal vertex degree)
 and by $D_\V=\max_{v\in\V}deg(v)$ 
(maximal vertex degree) of $\H$. Similarly we define the minimal 
edge degree $d_\E$ and the maximal edge
degree $D_\E$. 
The following simple lemma was found in 1971 and published in larger contexts in  \cite{A} (see also \cite{A2}).

{\it Covering Lemma 1}: { For every hypergraph $(\V,\E)$  there exists
a covering (of the vertices by an edge set) $\C\subset \E$ with }
\begin{equation}
|\C|\leq \frac{|\E|}{d_\V}\log |\V|.
\end{equation}

For most parameters a slightly better result was published in \cite{J},\cite{St}, and \cite{L}.

{\it Covering Lemma 2}:
{ For every hypergraph $(\V,\E)$   there exists a 
covering  of edges (by a vertex set) $C\subset \V$ with}
\begin{equation} 
|C|\leq \frac{|\V|}{d_\E}(1+\ln D_\V).
\end{equation}
These  resuts can be applied to our problems.

{\it $(r,s)$--sets or $s$--wise intersecting codes}: 

We apply Covering Lemma 2.
The vertex set $\V$ is the set of 
nonzero vectors in $\FF^r$
and the edge set $\E$ is the set of all $(r-s)$--flats.
The number of all $(r-s)$--flats is $(2^s-1)\bey{r}{r-s}$.
Thus, we have a regular uniform hypergraph with
 $|\V|=2^{r}-1$ and $|\E|=(2^s-1)\bey{r}{r-s}$.
Each $(r-s)$--flat has size $2^{r-s}$, that is $d_\E=2^{r-s}$.
 The number of 
$(r-s)$--flats in $\FF^r_2$
containing a given vector is $2^{r-s}\bey{r-1}{s-1}$. 
 Thus, the vertex degree
 is $d_\V= 2^{r-s}\bey{r-1}{s-1}$. In view of the lemma there is
a covering  $C$ with
$$
|C|\leq \frac{2^r-1}{2^{r-s}}\Bigl(1+\ln\Bigl(2^{r-s}\bey{r-1}{s-1}\Bigr)\Bigr)<$$
$$2^s(1+(r-s)s\ln 2+2\ln 2).$$
\begin{coro} \label{coro8} For integers $2\leq s\leq r$ we have
 \begin{equation}
  G_1(r,s)<2^s(s(r-s)\ln2+2\ln2+1).
 \end{equation}
\end{coro}

Recall that the upper bound in Theorem 4 is approximately $2^s\ln2(rs-\log s!)$.

Next we show that there are "good"  $(r,s)$--sets with an interesting structure:
 a union of $s$--subspaces of $\FF^r$. To this end we need the following simple fact.

\begin{lemm}\label{lemm4}
 A set of vectors $A\subset \FF^r$ is $(r,s)$--set if for every $(r-s)$--space $V\subset\FF^r$ 
there exists an $s$--space $U\subset  A$ such that $V\cap U={\bf 0}$.
\end{lemm}

\begin{proof} The proof is straightforward. Given an $(r-s)$--space $V$,
 the fact that the direct sum $V+U=\FF^r$  implies that
$U$ hits every coset of $V$.
\end{proof}

Consider a bipartite graph $\G=(\U\cup \V,\E)$ with bipartition $\U\cup \V$.
Define $\V$ to be  the set of all $s$--subspaces, and $\V$ to be the set of all
$(r-s)$--subspaces of $\FF^r$. 
Thus $|\U|=|\W|=\bey rs$. For $U\in\U$ and $V\in\V$ we have an edge $(U,V)\in \E$ if
and only if
$U\cap V={\bf 0}$. It is easy to see that given an $s$--subspace $U$, 
the number of $(r-s)$--subspaces avoiding $U$ is $2^{s(r-s)}$. Hence,
the degree of every vertex in $\G$ is $2^{s(r-s)}$. 

The problem now is to find a minimal cover $C\subset \U$ of the vertices  $\V$.
This clearly gives us an $(r,s)$--set.
 
Every hypergraph can be represented as a bipartite graph (or an incidence matrix)
and vice versa. Given a bipartite graph  $\G=(\U\cup \V,\E)$, let $d_\V$ be
the minimal degree of $\V$ and let $D_\U$ be the maximal degree of $\U$. 

The bipartite graph version of the Covering Lemma 2 is as follows. 
There exists a covering $C\subset \U$ of $\V$ with
\begin{equation}
 |C|\leq \frac{|\U|}{d_\V}(1+\ln D_\U).
\end{equation}
Applying this to our problem we get 
$$
|C|\leq \frac{\bey rs}{2^{s(r-s)}}(1+\ln 2^{s(r-s)})<4(1+s(r-s)\ln2).
$$
 This yields the following result.
\begin{theo}\label{theo14}
 There exists a $(k,s)$--set (resp. an $s$--wise intersecting $[n,k]$ code) consisting 
(resp. with   a generator matrix whose columns consist) of a union of 
less than $4(s(k-s)\ln2+1)$  subspaces of dimension $s$.
\end{theo}

{\it Generic erasure corecting sets}:

The vertex set $\V$ our hypergraph $(\V,\E)$ is the set of 
nonzero vectors in $\FF^r$.
A subset $E\subset \V$ is an edge in $\E$ if and only if $E$ is a
union of $s$ linearly independent cosets (defined in the proof of Theorem 13) of an $(r-s)$--subspace.
 Thus, the degree   of each edge is $s2^{r-s}$.
Furthermore, the degree of each vertex is
$\bey{r-1}{s-1}\prod^{s-1}_{i=1}(2^s-2^i)/(s-1)!.$

It is clear that a minimal edge covering $C$ gives 
an optimal generic erasure correcting $(r,s)$--set, that is $|C|=F(r,s)$.
Applying now   (V.2) we get
$$
F(r,s)=|C|\leq \frac{2^r-1}{s2^{r-s}}\Bigl(1+\ln \frac{\prod^{s-1}_{i=1}(2^s-2^i)\bey{r-1}{s-1}}{(s-1)!}\Bigr)<$$
$$2^s(r\ln 2-\ln s).$$

{\it Stopping redundancy of a binary linear code}:

Let $\C$ be an $[n,k,d]$ code and $\C^\bot$ be its dual code. Let also $r=n-k$ and $s=d-1$.
The vertex set $\V$ of our hypergraph is the set of all nonzero vectors of $\C$.
Given a set of coordinates $K\subset [n]$ with $|K|\leq s$, let $\C^\bot_{K}$ be the set  of all
vectors in $\C^\bot$ which have weight one in $K$. Note that $|\C^\bot_K|=|K|2^{r-|K|}\geq s2^{r-s} $.
Our edge set is defined as $\E=\{\C^\bot_K: K\subset[n], 1\leq |K|\leq s\}$.
Let $C\subset \V$ be a minimum vertex cover
 of the hypergraph $(\V,\E)$. It is easy to see that if $C$ is a
parity check matrix, that is $span(C)=\C^\bot$,
then $\rho(\C)=|C|$. Note that dim span$(C)\geq s$. Therefore, adding
 at most $r-s$ independent vectors to $C$ we get a parity check matrix. Thus, we
have $\rho(\C)\leq |C|+r-s$. 
Observe now that a vector ${\bf u}\in \C^\bot$ of weight $wt({\bf u})$ covers 
 $\alpha({\bf u})=wt({\bf u})\sum^{s}_{i=1}\binom{n-wt(\bf u)}{i-1}$ edges.
Let $t=wt({\bf u})$ be the weight for which  $\alpha({\bf u})$
is maximal over all choices of ${\bf u}\in \C^\bot$.
 Thus,
$(\V,\E)$ is a  hypergraph with the minimal edge degree $d_\E=s2^{r-s}$ 
and maximal vertex degree $\D_\V=t\sum^s_{i=1}\binom{n-t}{i-1}$.
Therefore, applying (V.2) we get
$$
|C| <\frac{2^r-1}{s2^{r-s}}\Bigl(1+\ln\Bigl(t\sum^s_{i=1}\binom{n-t}{i-1}\Bigr)\Bigr)
<$$
$$\frac{2^s}s\Bigl(1+\ln\sum_{i=1}^s\binom ni\Bigl).
$$

\begin{coro} \label{coro9}
 For an $[n,k,d]$ code $\C$ with $d\geq 3$ we have
$$
 \rho(\C)<\frac{2^{d-1}}{d-1}\Bigl(1+\ln\sum_{i=1}^{d-1}\binom ni\Bigl)+n-k-d+1.
$$
\end{coro}

 Notice that although we do not always get the best known constants,
however we achieve the same order of magnitude  for the upper bounds. 
Since this simple approach gives almost the same results as those of presented before,
it should be followed further by finding better covering results
using for example Maximal Code Lemma (\cite{A1}, p.238)
or ideas and methods described in (\cite{AB}, ch.3).

{\bf Acknowledgement} The second author would like to thank anonymous referees for their   comments.

\end{document}